\documentclass{amsart}
\usepackage{amsmath}
\usepackage{amsthm,amsfonts,amssymb,array,float}
\usepackage[norelsize,lined,boxed,commentsnumbered,linesnumbered]{algorithm2e}
\usepackage{tikz}
\usetikzlibrary{shadows}
\usepackage{pgfplots} 

\newcounter{indice}

\newcommand{\permutation}[1]{
\setcounter{indice}{0};
\foreach \i in {#1} 
\addtocounter{indice}{1};

\addtocounter{indice}{1}
\draw [help lines] (1,1) grid (\theindice,\theindice);

\setcounter{indice}{1};

\foreach \i in { #1 } {
\draw (\theindice+.5,\i+.5) [fill] circle (.2);
\addtocounter{indice}{1};
}
\addtocounter{indice}{-1};
}

\newtheorem{prop}{Proposition}
\newtheorem{defn}{Definition}
\newtheorem{thm}{Theorem}

\newcommand{\N}{\ensuremath{\mathbb N}}
\begin{document}
\title{Simple permutations poset}

\author{Adeline Pierrot}
\email{adeline.pierrot@liafa.jussieu.fr}
\address{LIAFA, UMR 7089, Universit\'e Paris Diderot, Paris, France}

\author{Dominique Rossin}
\email{rossin@lix.polytechnique.fr}
\address{LIX, UMR 7161, Ecole Polytechnique, Palaiseau, France}

\begin{abstract}
This article studies the poset of simple permutations with respect to the pattern involvement. We specify results on critically indecomposable posets obtained by  Schmerl and Trotter in \cite{SchmerlTrotter1993} to simple permutations and prove that if $\sigma, \pi$ are two simple permutations such that $\pi < \sigma$ then there exists a chain of simple permutations $\sigma^{(0)} = \sigma, \sigma^{(1)}, \ldots, \sigma^{(k)}=\pi$ such that $|\sigma^{(i)}| - |\sigma^{(i+1)}| = 1$ - or $2$ when permutations are exceptional- and $\sigma^{(i+1)} < \sigma^{(i)}$. This characterization induces an algorithm polynomial in the size of the output to compute the simple permutations in a wreath-closed permutation class.
\end{abstract}
\maketitle
\section{Introduction}

Simple permutations are the permutations which map no proper non-singleton interval onto an interval. Those permutations play a key role in the study of permutation classes, that is closed sets of permutations. More precisely,  they are core objects in the substitution decomposition of permutations. For example, if a class contains a finite number of simple permutations then it is finitely based, meaning that the class can be expressed as the set of permutations that do not contain as pattern any permutation of a finite set $B$. Moreover, the generating function of the permutation class is algebraic. On second hand, even if the pattern involvement problem is NP-complete in general, there exists a FPT algorithm \cite{BR06} where the parameter is the length of the largest simple permutation that appears as a pattern of the involved permutations.

In this article, we study the set of simple permutations with respect to the pattern containment relation. 
In \cite{SchmerlTrotter1993}, Schmerl and Trotter  study critically indecomposable partially ordered sets of integers and prove many structural results. 
They notice that their results still holds for all relational structures, in our case, permutations.

In this article, we  focus on simple permutations and show that the general results on integers can be refined in our case.
More precisely, if $\sigma, \pi$ are two simple permutations such that $\pi$ is pattern of $\sigma$ -$\pi < \sigma$- \cite{SchmerlTrotter1993} proves that there exists a chain of simple permutations $\sigma^{(0)} = \sigma, \sigma^{(1)}, \ldots, \sigma^{(k)}=\pi$ such that $|\sigma^{(i)}| - |\sigma^{(i+1)}| = 1$ or $2$  and $\sigma^{(i+1)} \prec \sigma^{(i)}$.
Using the structure of simple permutations, we strengthen  the result and show that in the case of permutations, we can find a chain with all size differences of $1$ except when $\sigma$ is exceptional. In the latter there exists a chain with all size differences of $2$.

This structural result on permutations and pattern involvment has many consequences. First, it allows us to compute the average number of points in a simple permutations that can be removed -one at each time- in order to obtain another simple permutation.
On another hand, it gives rise to a polynomial algorithm to generate simple permutations of a wreath-closed class of permutations.  This algorithm roughly starts by looking to simple permutations of size $4$ and iterates over the size of permutations. The characterization of the preceding chain translates into a polynomial time algorithm for finding simple permutations of size $n+1$ in $Av(B)$ knowing only simple permutations of size $n$ and $n-1$ in this class. Note that our algorithm requires no pattern involvement test. Although we give this algorithm in the general framework of wreath-closed permutation classes, we apply it for classes containing only a finite number of simple permutations. It can be also used to generate the simple permutations of any wreath-closed class up to a given size, even if the number of simple permutations is infinite.
 
 Together with preceding results of Albert, Atkinson, Brignall \cite{AA05, Bri06, Bri07, BHV06b} and our results \cite{BBR09, BBPR09}, this algorithm allows
 to compute the generating function of a wreath-closed class of permutations containing a finite number of simple permutations. The overall complexity of this algorithm is polynomial in the number of simple permutations. Some statistical results are given in the last section.

 \section{Definitions}

A permutation $\sigma$ is a bijective map from $\{1\ldots n\}$ onto $\{ 1 \ldots n\}$ with $n = |\sigma|$. We either represent a permutation by a word $\sigma = \sigma_1\sigma_2\ldots \sigma_n$ where $\sigma_i = \sigma(i)$ or its graphical representation in a grid (see Figure~\ref{fig:exceptional} for examples).

Let $\pi$ and $\sigma$ be two permutations. We say that $\pi = \pi_1 \ldots \pi_k$ is a {\em pattern} of $\sigma$ and we write $\pi \preceq \sigma$ if and only if there exist $i_1 < i_2 < \ldots < i_k$ such that $\pi$ is order-isomorphic to $\sigma_{i_1} \sigma_{i_2} \ldots \sigma_{i_k}$. A permutation class is a downward-closed set of permutations under pattern relation.
We can also define permutation classes by the set of minimal (for $\preceq$)  permutations not in the class which is an antichain for the pattern relation. This set is called the {\em basis} of the class. We denote by $Av(B)$ the permutation class which is the set of permutations that do not contain any of the permutations $\pi \in B$ as a pattern. For example $Av(231)$ is the class of one-stack sortable permutations.
When the basis $B$ contains only simple permutations the permutation class $Av(B)$ is said to be {\em wreath-closed}. Wreath-closed classes are defined in \cite{AA05} in a
different way but the authors prove that this definition is equivalent.

An interval in a permutation is a consecutive set of elements $\sigma_i \sigma_{i+1}\ldots \sigma_j$ such that the set of values $\{ \sigma_i, \sigma_{i+1},\ldots \sigma_j \}$ is an interval. A permutation is said to be simple if and only if its intervals are trivial -the singletons and the whole permutation-. As example 1, 12, 21, 2413 and 3142 are the simple permutations of size $\leq 4$. A subset of simple permutations, called exceptional ones plays a key role in this article.

\begin{defn}
Exceptional permutations are permutations defined below for every $m \geq 2$ (see Figure~\ref{fig:exceptional}):
\begin{itemize}
\item $2\ 4\ 6\ 8\ldots (2m)\ 1\ 3\ 5\ldots (2m-1)$ --- type 1
\item $(2m-1)\ (2m-3)\ldots 1\ (2m)\ (2m-2)\ldots 2$ --- type 2
\item $(m+1)\ 1\ (m+2)\ 2\ldots (2m)\ m$ --- type 3
\item $m\ (2m)\ (m-1)\ (2m-1)\ldots 1\ (m+1)$ --- type 4
\end{itemize}
\end{defn}

\begin{figure}[ht]
\begin{center} 
\begin{tikzpicture}[scale=.2]
\permutation{2,4,6,8,10,1,3,5,7,9}
\end{tikzpicture}
\hspace{0.5cm}
\begin{tikzpicture}[scale=.2]
\permutation{9,7,5,3,1,10,8,6,4,2}
\end{tikzpicture}
\hspace{0.5cm}
\begin{tikzpicture}[scale=.2]
\permutation{6,1,7,2,8,3,9,4,10,5}
\end{tikzpicture}
\hspace{0.5cm}
\begin{tikzpicture}[scale=.2]
\permutation{5,10,4,9,3,8,2,7,1,6}
\end{tikzpicture}
\caption{Exceptional permutation of type 1, 2, 3 and 4}
\label{fig:exceptional}
\end{center}
\end{figure}
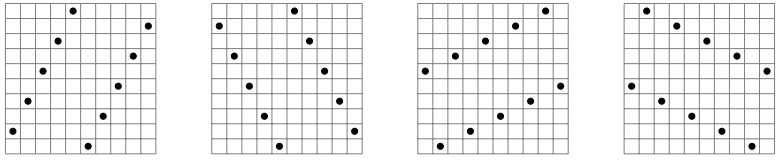
Notice that exceptional permutations are simple.
Notice also that, if we remove the symbols $2m-1$ and $2m$ from the first two types, we obtain an other exceptional permutation of the same type; and likewise if we remove the symbols in the last two positions from the types 3 and 4 and renormalize the result. 

\begin{prop}\label{prop:TypeExceptional}
Let $\sigma, \sigma'$ be two exceptional permutations with $|\sigma| \leq |\sigma'|$. Then $\sigma \preceq \sigma'$ if and only if $\sigma$ and $\sigma'$ are exceptional permutations of the same type.
\end{prop}

\begin{figure}[ht]
\begin{minipage}[b]{.5\linewidth}
\begin{center}
\begin{tikzpicture}
\begin{scope}[scale=.2]
\permutation{1,3,5,7,9,11,2,4,6,8,10}
\end{scope}
\begin{scope}[xshift=3cm,scale=.2]
\permutation{6,11,5,10,4,9,3,8,2,7,1}
\end{scope}
\end{tikzpicture}
\caption{Parallel alternations}
\label{fig:para}
\end{center}
\end{minipage}
\begin{minipage}[b]{.5\linewidth}
\begin{center}
\begin{tikzpicture}
\begin{scope}[scale=.2]
\permutation{6,5,7,4,8,3,9,2,10,1,11}
\end{scope}
\begin{scope}[xshift=3cm,scale=.2]
\permutation{1,3,5,7,9,11,10,8,6,4,2}
\end{scope}
\end{tikzpicture}
\caption{Wedge alternations}
\label{fig:alternations}
\end{center}
\end{minipage}
\end{figure} 
A more general kind of permutations containing exceptionnal permutations will appear naturally in this article.
An {\it alternation} is a permutation in which every odd entry lies to the left of every even entry, or any symmetry of such a permutation. A {\it parallel} alternation is one in which these two sets of entries form monotone subsequences, either both increasing or both decreasing. A {\it wedge} alternation is one in which the two sets of entries form monotone subsequences pointing in opposite directions. See Figures~\ref{fig:para} and \ref{fig:alternations} for examples. Among parallel or wedge alternations, only exceptional permutations are simple.


\section{Pattern containment on simple permutations}\label{sec:MotifsSimples}

\subsection{Simple patterns of simple permutations}

In this section, we study the poset of simple permutations with respect to the pattern relation.  More precisely, we specify and extend results from \cite{SchmerlTrotter1993} to permutations.
To each permutation $\sigma$ is associated a poset $P( \sigma )=([1..n],\prec)$ where $i \prec  j \Leftrightarrow (i < j \text{ and } \sigma_i < \sigma_j)$. 
A poset on $[1..n]$ is called {\em indecomposable} \cite{SchmerlTrotter1993} if  it does not contain any non-trivial interval with respect to the relations , $\prec$ in our case. It is called {\em critically indecomposable} if furthermore whenever an element is removed, the resulting poset is not indecomposable.
In the specific case of  permutations those poset characteristics can be  translated as stated in the following propositions:

\begin{prop}
The poset $P(\sigma)$ is indecomposable if and only if $\sigma$ is simple.
\end{prop}

\begin{prop}\label{prop:critic}
$P(\sigma)$ is critically indecomposable if and only if $\sigma$ is exceptional. 
\end{prop}

\begin{proof}
This is a consequence of  Corollary  $5.8 (2)$ of \cite{SchmerlTrotter1993}.
\end{proof}

The standard poset isomorphism can be transposed  to integer posets. 

\begin{defn}
Let  $A,B \subset \N$ posets, then $A \equiv B$ if  $A \sim B$  as posets and the isomorphism keeps the natural integer ordering on  $\N$.
\end{defn}

\begin{prop}
Let $\sigma$ and $\pi$ be permutations. Then $\pi \preceq \sigma$ if and only if there exists $A \subset P(\sigma)$ such that $A \equiv P(\pi)$.
\end{prop}

The next four propositions are mere translations of results from \cite{SchmerlTrotter1993} on permutations.

\begin{prop}\label{prop:size4}
Let $\sigma$ be a simple permutation  with $|\sigma| \geq 3$, then either $2\,4\,1\,3$ or $3\,1\,4\,2$ is a simple pattern of $\sigma$.
\end{prop}

\begin{proof}
This is the mere translation of Theorem $2.1$ of \cite{SchmerlTrotter1993}, noticing that there are no simple permutations of size $3$ and there are only two simple permutations of size $4$.
\end{proof}

\begin{prop}[Theorem $2.2$  of \cite{SchmerlTrotter1993}]\label{prop:proposition0.2}
Let $\pi, \sigma$ be two simple permutations with $\pi \preceq \sigma$. If $3 \leq |\pi| \leq |\sigma|-2$, then there exists 
a simple permutation $\tau$ such that $\pi \preceq \tau \preceq \sigma$ and $|\tau| = |\pi|+2$.
\end{prop}

\begin{prop}\label{prop:exceptional}
Let $\sigma$ be an exceptional  permutation. If  $3 \leq m \leq |\sigma|$, then $\sigma$ has a simple pattern of size $m$ if and only if $m$ is even.
\end{prop}

\begin{proof}
This is a direct consequence of Corollary $3.1$ of \cite{SchmerlTrotter1993} as every exceptional permutation is of even size.
\end{proof}

\begin{prop}[Corollary $5.10$ of \cite{SchmerlTrotter1993} + Proposition~\ref{prop:size4}]\label{prop:simplePattern}
Let $\sigma$ be a non exceptional simple permutation. If $4 \leq m \leq |\sigma|$ then $\sigma$ has a simple pattern of size $m$.
\end{prop}

The preceding result holds for non exceptional permutations. For exceptional ones, the following proposition concludes:

\begin{prop} \label{prop:except}
If $\sigma$ is an exceptional permutation, then for every $m$ such that $3 \leq m \leq |\sigma|$:
\begin{itemize}
\item If $m$ is odd, then $\sigma$ has no simple pattern of size $m$.
\item Otherwise $m$ is even and $\sigma$ has exactly one simple pattern of size $m$ which is the exceptional permutation of the same type as $\sigma$.
\end{itemize}
\end{prop}

\begin{proof}
The first item -$m$ is odd- is a direct consequence of Proposition~\ref{prop:exceptional}. For the second point, $\sigma$ has at least one simple pattern $\pi$ of size $m$ by Proposition~\ref{prop:exceptional}. Suppose now that $\pi$ is not exceptional, then $m \geq 5$ as simple permutations of size $4$ are exceptional. Then, using Proposition~\ref{prop:simplePattern}, $\pi$ has a simple pattern $\tau$ of size $5$, thus $\tau$ is a pattern of $\sigma$ but of odd size which is forbidden by Proposition~\ref{prop:exceptional}. So $\pi$ is exceptional and of the same type as $\sigma$ from Proposition~\ref{prop:TypeExceptional}.
\end{proof}

A direct consequence of the preceding proposition is that all simple patterns of exceptional permutations are exceptional. For non exceptional ones the following proposition describes the pattern containment relation:

\begin{prop} \label{prop:size-1}
Let $\sigma$ be a simple permutation of size $\geq 5$. Then $\sigma$ has a simple pattern of size $|\sigma|-1$ if and only if $\sigma$ is not exceptional.
\end{prop}
\begin{proof}
Consequence of Proposition~\ref{prop:exceptional} and Proposition~\ref{prop:simplePattern} above.
\end{proof}

\subsection{Simple pattern containing a given simple permutation}

The results obtained in the preceding section describe how a simple permutation can give other simple permutations by deleting elements. In the sequel, we add another constraint on patterns, that is we want to delete elements in a simple permutation $\sigma$ containing a simple permutation $\pi$ as a pattern only by deleting one or two elements to obtain another simple permutation $\sigma'$ such that $\pi \preceq \sigma'$.

Theorem~\ref{thm:theorem0.3} deals with the case $|\pi| = |\sigma|-2$ and relies on the following intermediate result:

\begin{prop}\label{prop:intervalleOuCoin}
Let $\tau$ be a non simple permutation such that $\tau \setminus \{\tau_i\}$ is simple. Then $\tau_i$ belongs to an interval of size $2$ of $\tau$ or is in a corner of the graphical representation of $\tau$.
\end{prop}

\begin{proof}
As $\tau$ is not simple, $\tau$ contains at least one non-trivial interval $I$. As $I$ is an interval of $\tau$, $I \setminus \{\tau_i\}$ is an interval of $\tau \setminus \{\tau_i\}$. But $\tau \setminus \{\tau_i\}$ is simple thus $I \setminus \{\tau_i\}$ is a trivial interval of $\tau \setminus \{\tau_i\}$, hence is a singleton $\{\tau_k\}$ or the whole permutation $\tau \setminus \{\tau_i\}$. But $I$ is non-trivial so in the first case $I = \{\tau_i, \tau_k\}$ and $\tau_i$ belongs to an interval of size $2$ of $\tau$, and in the second case $I = \tau \setminus \{\tau_i\}$ and $\tau_i$ is in a corner of the graphical representation of $\tau$.
\end{proof}

\begin{thm}\label{thm:theorem0.3}
Let $\sigma = \sigma_1 \sigma_2 \ldots \sigma_n$ be a non exceptional simple permutation of size $n \geq 4$ and $\pi$ a simple permutation of size $n-2$ such that $\pi \preceq \sigma$. Then there exists a simple permutation $\tau$ of size  $n-1$ such that $\pi \preceq \tau \preceq \sigma$.
\end{thm} 

\begin{proof}
Supppose that such a permutation $\tau$ does not exist. We prove that this leads to a contradiction. Let $i,j$ such that $\pi = \sigma \setminus \{\sigma_i, \sigma_j\}$. If $\sigma \setminus \{ \sigma_i \}$ is simple then $\tau = \sigma \setminus \{ \sigma_i \}$ would contradict our hypothesis. Thus $\sigma \setminus \{ \sigma_i \}$ is not simple, but $\pi = \sigma \setminus \{\sigma_i, \sigma_j\}$ is simple. From Proposition~\ref{prop:intervalleOuCoin} $\sigma_j$ belongs to an interval of size $2$ of $\sigma \setminus \{ \sigma_i \}$ or is in a corner of the bounding box of the graphical representation of $\sigma \setminus \{ \sigma_i \}$ thanks to $\pi$. By symmetry between $i$ and $j$ the same results holds when exchanging these two indices. So there are $3$ different cases:
\begin{itemize}
\item $\sigma_i$ and $\sigma_j$ are both in a corner thanks to $\pi$. In that case $\pi$ is a non trivial interval of $\sigma$, which contradicts the fact that $\sigma$ is simple.
\item $\sigma_i$ belongs to an interval $I$ of size $2$ of $\sigma \setminus \{ \sigma_j \}$ and $\sigma_j$ is in a corner of $\sigma \setminus \{ \sigma_i \}$ thanks to $\pi$ (the same proof holds when exchanging $i$ and $j$). $\sigma$ is simple thus $\sigma_j$ is not in a corner of $\sigma$, but is in a corner of $\sigma \setminus \{ \sigma_i \}$ thus $\sigma_i$ is the only point separating $\sigma_j$ from a corner (see Figure~\ref{fig:cas2} for an example). Let $i_1$ such that $I = \{i,i_1\}$, then $\sigma_j$ is the only point separating $\sigma_{i_1}$ from $\sigma_i$, so $\pi = \sigma \setminus \{  \sigma_{i_1},\sigma_j \}$. If $\sigma \setminus \{ \sigma_{i_1} \}$ is simple then $\tau = \sigma \setminus \{ \sigma_{i_1} \}$ would answer the theorem, contradicting our hypothesis. 
Thus $\sigma \setminus \{ \sigma_{i_1} \}$ is not simple but $\pi = \sigma \setminus \{\sigma_{i_1}, \sigma_j\}$ is simple, hence from Proposition~\ref{prop:intervalleOuCoin} $\sigma_j$ belongs to an interval $J$ of size $2$ of $\sigma \setminus \{ \sigma_{i_1} \}$ or is in a corner of $\sigma \setminus \{ \sigma_{i_1} \}$ which is impossible as $\sigma_i$ separate it from one corner and $|\sigma| \geq 4$. Let $j_1$ such that $J = \{ j,j_1 \}$, then $\pi = \sigma \setminus \{ \sigma_{i_1}, \sigma_{j_1} \}$. If $\sigma \setminus \{ \sigma_{j_1} \}$ is simple then $\tau = \sigma \setminus \{ \sigma_{i_1} \}$ answer the theorem, contradiction.
Let $i_0 = i$ and $j_0 = j$, we recursively build $i_0,j_0,i_1,j_1,\ldots$ such that $\forall k, \pi = \sigma \setminus \{ \sigma_{i_k},\sigma_{j_k} \} = \sigma \setminus \{ \sigma_{j_k}, \sigma_{i_{k+1}} \}$ and $\sigma \setminus \sigma_{i_k}$ and $\sigma \setminus \sigma_{j_k}$ are not simple, until reaching all points of $\sigma$. 
$\sigma_i$ and $\sigma_{i_1}$ are in increasing -or decreasing- order. For each case, -see Figure~\ref{fig:cas2}-, $\sigma_{i_1}$ has a determined position. Then positions of $\sigma_{i_k}$ and $\sigma_{j_k}$ are fixed for all $k$ as $\sigma_{i_k}$ does not separate $\sigma_{i_{k-1}}$ from $\sigma_{i_{k-2}}$. 
Depending of the position of $\sigma_{i_1}$, $\sigma$ is either a parallel alternation or a wedge alternation thus is exceptional or not simple, a contradiction.
\begin{figure}[ht]
\begin{center} 
\begin{tikzpicture}[scale=.3]
\useasboundingbox (-1,-1) rectangle (10,8);
\draw (1.5,0.5) [fill] circle (0.2);
\draw (0.5,4.5) [fill] circle (0.2);
\node at (1.5,-1) {$\sigma_j$};
\node at (-1,4.5) {$\sigma_i$};
\draw [thick] (2,1) rectangle node {$\pi$} +(6.3,7);
\draw [thick] (0,0) rectangle (8.3,8);
\draw (1,0) -- (1,8);
\end{tikzpicture}
\begin{tikzpicture}[scale=.3]
\useasboundingbox (-1,-1) rectangle (10,8);
\draw [help lines] (1,0) grid (4,2);
\draw [help lines] (0,4) grid (3,6);
\draw (1.5,0.5) [fill] circle (0.2);
\draw (3.5,1.5) [fill] circle (0.2);
\draw (0.5,4.5) [fill] circle (0.2);
\draw [fill] (2.5,5.5) circle (0.2);
\node at (1.5,-1) {$\sigma_j$};
\node at (5,1.7) {$\sigma_{j_1}$};
\node at (-1,4) {$\sigma_i$};
\node at (4,5.6) {$\sigma_{i_1}$};
\draw [thick] (2,1) rectangle node {$\pi$} +(6.3,7);
\draw [thick] (0,0) rectangle (8.3,8);
\end{tikzpicture}
\begin{tikzpicture}[scale=.3]
\useasboundingbox (-1,-1) rectangle (10,8);
\draw [help lines] (1,0) grid (4,2);
\draw [help lines] (0,4) grid (3,6);
\draw (1.5,0.5) [fill] circle (0.2);
\draw (3.5,1.5) [fill] circle (0.2);
\draw (0.5,5.5) [fill] circle (0.2);
\draw [fill] (2.5,4.5) circle (0.2);
\node at (1.5,-1) {$\sigma_j$};
\node at (5,1.7) {$\sigma_{j_1}$};
\node at (-1,5.5) {$\sigma_i$};
\node at (4,3.9) {$\sigma_{i_1}$};
\draw [thick] (2,1) rectangle node {$\pi$} +(6.3,7);
\draw [thick] (0,0) rectangle (8.3,8);
\end{tikzpicture}
\caption{Graphical representation of $\sigma$ in the case 2.}\label{fig:cas2}
\end{center}
\end{figure}
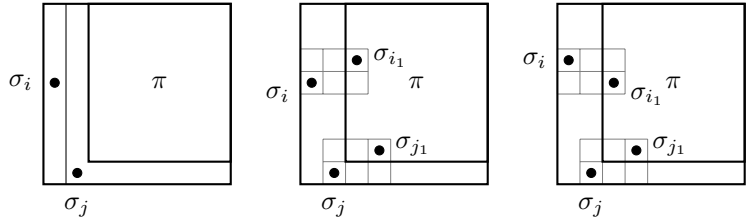
\item $\sigma_i$ belongs to an interval $I$ of size $2$ of $\sigma \setminus \{\sigma_{j}\}$ and $\sigma_j$ belongs to an interval $J$ of size $2$ of $\sigma \setminus \{\sigma_{i}\}$. Let $i_1$ such that $I = \{ i,i_1 \}$, $I$ is an interval of $\sigma \setminus \{\sigma_{j}\}$ but $\sigma$ is simple thus $\sigma_j$ is the only point separating $\sigma_i$ from $\sigma_{i_1}$.
Let $j_1$ such that $J = \{ j,j_1 \}$, $J$ is an interval of $\sigma \setminus \{\sigma_{i}\}$ but $\sigma$ is simple thus $\sigma_i$ is the only point separating $\sigma_j$ from $\sigma_{j_1}$. This indeed is one of the two cases depicted in Figure~\ref{fig:cas3} (up to symmetry). But $\pi = \sigma \setminus \{ \sigma_i, \sigma_{j} \} = \sigma \setminus \{ \sigma_{i_1}, \sigma_{j} \}$. If $\sigma \setminus \{ \sigma_{i_1} \}$ is simple then $\tau = \sigma \setminus \{ \sigma_{i_1} \}$ answer the theorem, a contradiction. Thus $\sigma \setminus \{ \sigma_{i_1} \}$ is not simple but $\pi = \sigma \setminus \{ \sigma_{i_1}, \sigma_{j} \}$ is simple thus from Proposition~\ref{prop:intervalleOuCoin} $\sigma_j$ belongs to an interval $J'$ of size $2$ of $\sigma \setminus \{ \sigma_{i_1}  \}$ or lies in a corner of $\sigma \setminus \{ \sigma_{i_1} \}$ which is impossible if $i_1 \neq n$ (up to symmetry). Let $J' = \{ j, j'\}$ then $\pi = \sigma \setminus \{ \sigma_{i_1},\sigma_j \} = \sigma \setminus \{ \sigma_{i_1},\sigma_{j'} \}$. 
If $\sigma \setminus \{ \sigma_{j'} \}$ is simple then $\tau = \sigma \setminus \{ \sigma_{i_1}\}$ answer the theorem, contradiction.
Moreover $\pi = \sigma \setminus \{ \sigma_{i},\sigma_j \} = \sigma \setminus \{ \sigma_{i},\sigma_{j_1} \}$. If $\sigma \setminus \{ \sigma_{j_1} \}$ 
is simple then $\tau = \sigma \setminus \{ \sigma_{j_1}\}$ fulfil the theorem, contradiction. Thus $\sigma \setminus \{ \sigma_{j_1}\}$ is not simple but $\pi = \sigma \setminus \{ \sigma_{j_1},\sigma_i \}$ is simple so that $\sigma_i$ belongs to an interval $I'$ of size $2$ of $\sigma \setminus \{ \sigma_{j_1} \}$ or lies in a corner of $\sigma \setminus \{ \sigma_{j_1}\}$, which is impossible if $j_1 \neq 1$ (up to symmetry). Let $i'$ such that $I' = \{ i,i'\}$, then $\pi = \sigma \setminus \{ \sigma_{i},\sigma_{j_1} \} = \sigma \setminus \{ \sigma_{i'},\sigma_{j_1} \}$. If $\sigma \setminus \{ \sigma_{i'} \}$ is simple then $\tau = \sigma \setminus \{ \sigma_{i'}\}$ fulfil our theorem, contradiction.
Similarly to the preceding case, if $i_0 = i, j_0 = j$ then we can prove by induction until reaching all points of $\sigma$ that either $\sigma$ is a parallel alternation or a wedge permutation so that $\sigma$ is exceptional or not simple which leads to a contradiction.
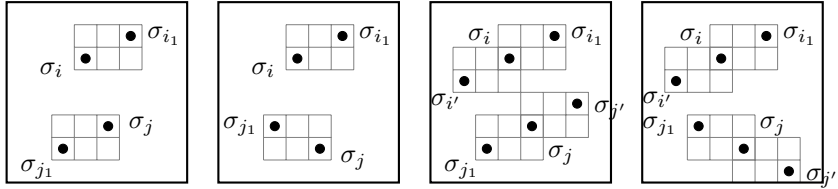
\begin{figure}[ht]
\begin{center}
\begin{tikzpicture}[scale=.3]
\useasboundingbox (0,0) rectangle (9,9);
\draw [help lines] (2,1) grid (5,3);
\draw [help lines] (3,5) grid (6,7);
\draw (2.5,1.5) [fill] circle (0.2);
\draw (4.5,2.5) [fill] circle (0.2);
\draw (3.5,5.5) [fill] circle (0.2);
\draw [fill] (5.5,6.5) circle (0.2);
\node at (1.4,0.6) {$\sigma_{j_1}$};
\node at (6,2.5) {$\sigma_j$};
\node at (2,5) {$\sigma_i$};
\node at (7,6.5) {$\sigma_{i_1}$};
\draw [thick] (0,0) rectangle (8,8);
\end{tikzpicture}
\begin{tikzpicture}[scale=.3]
\useasboundingbox (0,0) rectangle (9,9);
\draw [help lines] (2,1) grid (5,3);
\draw [help lines] (3,5) grid (6,7);
\draw (2.5,2.5) [fill] circle (0.2);
\draw (4.5,1.5) [fill] circle (0.2);
\draw (3.5,5.5) [fill] circle (0.2);
\draw [fill] (5.5,6.5) circle (0.2);
\node at (1,2.5) {$\sigma_{j_1}$};
\node at (6,1) {$\sigma_j$};
\node at (2,5) {$\sigma_i$};
\node at (7,6.5) {$\sigma_{i_1}$};
\draw [thick] (0,0) rectangle (8,8);
\end{tikzpicture}
\begin{tikzpicture}[scale=.3]
\useasboundingbox (0,0) rectangle (9,9);
\draw [help lines] (2,1) grid (5,3);
\draw [help lines] (3,5) grid (6,7);
\draw [help lines] (4,2) grid (7,4);
\draw [help lines] (1,4) grid (4,6);
\draw (6.5,3.5) [fill] circle (0.2);
\draw (5.5,6.5) [fill] circle (0.2);
\draw (2.5,1.5) [fill] circle (0.2);
\draw (4.5,2.5) [fill] circle (0.2);
\draw (1.5,4.5) [fill] circle (0.2);
\draw (3.5,5.5) [fill] circle (0.2);
\node at (1.4,0.6) {$\sigma_{j_1}$};
\node at (5.8,1.2) {$\sigma_j$};
\node at (8,3.2) {$\sigma_{j'}$};
\node at (2.2,6.5) {$\sigma_i$};
\node at (0.7,3.6) {$\sigma_{i'}$};
\node at (7,6.5) {$\sigma_{i_1}$};
\draw [thick] (0,0) rectangle (8,8);
\end{tikzpicture}
\begin{tikzpicture}[scale=.3]
\useasboundingbox (0,0) rectangle (9,9);
\draw [help lines] (2,1) grid (5,3);
\draw [help lines] (3,5) grid (6,7);
\draw [help lines] (1,4) grid (4,6);
\draw [help lines] (4,0) grid (7,2);
\draw (1.5,4.5) [fill] circle (0.2);
\draw (2.5,2.5) [fill] circle (0.2);
\draw (4.5,1.5) [fill] circle (0.2);
\draw (3.5,5.5) [fill] circle (0.2);
\draw [fill] (5.5,6.5) circle (0.2);
\draw (6.5,0.5) [fill] circle (0.2);
\node at (0.8,2.5) {$\sigma_{j_1}$};
\node at (5.8,2.5) {$\sigma_j$};
\node at (8,0.2) {$\sigma_{j'}$};
\node at (0.7,3.7) {$\sigma_{i'}$};
\node at (2.2,6.5) {$\sigma_i$};
\node at (7,6.5) {$\sigma_{i_1}$};
\draw [thick] (0,0) rectangle (8,8);
\end{tikzpicture}
\caption{Graphical representation of $\sigma$ in the case 3.}\label{fig:cas3}
\end{center}
\end{figure}
\end{itemize}
\end{proof}

Thanks to Theorem~\ref{thm:theorem0.3} and a simple induction, we are able to state our main result on pattern involvement.

\begin{thm}\label{thm:main}
Let $\sigma \not= \pi$ be two simple permutations, $\sigma$ non exceptional. If $\pi \prec \sigma$ and $|\pi| \geq 3$ then there exists a simple permutation $\tau$ such that $\pi \preceq \tau \prec \sigma$ and $|\tau| = |\sigma|-1$.
\end{thm}
\begin{proof}
We prove this result by induction on $|\sigma|-|\pi|$ using Proposition~\ref{prop:proposition0.2}. If $|\sigma|-|\pi|$ is odd, using recursively Proposition~\ref{prop:proposition0.2} we find a simple permutation $\tau$ 
such that $\pi \preceq \tau \preceq \sigma$ and $|\tau| = |\sigma|-1$. If $|\sigma|-|\pi|$ is even, we find a simple permutation $\tau'$ such that $\pi \preceq \tau' \preceq \sigma$ and $|\tau'| = |\sigma|-2$ and we apply Theorem~\ref{thm:theorem0.3} which ensure the existence of a simple permutation $\tau$ such that $\pi \preceq \tau' \preceq \tau \preceq \sigma$ and $|\tau| = |\sigma|-1$.
\end{proof}


\section{Simple permutations poset}\label{sec:poset}

We study the poset of simple permutations of size $\geq 4$ with respect to the pattern containment relation. We can represent this poset by an oriented graph $G$, whose vertices are the simple permutations and there is an edge from a simple permutation $\sigma$ to a simple permutation $\pi$ if and only if $\pi \prec \sigma$ and there is no simple permutation $\tau$ such that $\pi \prec \tau \prec \sigma$. Then $\pi \prec \sigma$ if and only if there is a path from $\sigma$ to $\pi$ in $G$. From Theorem~\ref{thm:main}, if $\sigma$ is not exceptional there is an edge from $\sigma$ to $\pi$ if and only if we can obtain $\pi$ from $\sigma$ by deleting one point, and from Proposition~\ref{prop:except} if $\sigma$ is exceptional, there is an edge from $\sigma$ to $\pi$ if and only if $\pi$ is exceptional of the same type of $\sigma$ and $|\sigma| = |\pi|+2$. In this section we study other properties of $G$.

\subsection{Paths in the simple permutations poset}

In the next theorem, we prove that if a simple permutation $\sigma$ has a simple pattern $\pi$, then there is a path in $G$ from $\sigma$ to $\pi$ in the graph whose first part consists of non exceptional simple permutations of consecutive sizes and second part of exceptional permutations (one of the parts can be empty). From Proposition~\ref{prop:except} it is obvious that reciprocally, all paths from $\sigma$ to $\pi$ are of this form.
Then we extend this result to prove that whenever $\sigma$ is not exceptionnal, there is such a path such that the second part of the path is empty, that is we can reach $\pi$ from $\sigma$ by deleting one element at a time and all involved permutations are simple.

\begin{thm}
\label{thm:chain}
Let $\pi \neq \sigma$ be simple permutations. If $\pi \preceq \sigma$ and $|\pi| \geq 3$, then there exists a chain of simple permutations $\sigma^{(0)} = \sigma, \sigma^{(1)}, \ldots, \sigma^{(k-1)}, \sigma^{(k)}=\pi$ and $m \in \{0 \dots k\}$ such that $\sigma^{(i)} \preceq \sigma^{(i-1)}$ and:
\begin{itemize}
\item $|\sigma^{(i-1)}| - |\sigma^{(i)}| = 1$ if $1 \leq i \leq m$,  
\item $|\sigma^{(i-1)}| - |\sigma^{(i)}| = 2$ if $m+1 \leq i \leq k$ 
\item if $m < k$ then $\sigma^{(i)}$ is exceptional for $m \leq i \leq k$.
\end{itemize}
\end{thm}

\begin{proof}
If $\sigma$ is exceptional, then $\pi$ is exceptional of the same type as $\sigma$ (Proposition~\ref{prop:except}). Then we set $m = 0$ and $k = (|\sigma|-|\pi|)/2$, and $\sigma^{(i)}$ are exceptional permutations of the same type as $\sigma$ and size between $|\pi|$ and $|\sigma|$.

If $\sigma$ is not exceptional, we set $\sigma^{(0)} = \sigma$, and we construct $\sigma^{(i)}$ by induction while $\sigma^{(i-1)}$ is not exceptional and $\pi \neq \sigma^{(i-1)}$: from Theorem~\ref{thm:main}, there exists a simple permutation $\sigma^{(i)}$ such that $\pi \preceq \sigma^{(i)} \preceq \sigma^{(i-1)}$ and $|\sigma^{(i)}|=|\sigma^{(i-1)}|-1$. We iterate until $\sigma^{(j)} = \pi$, then $m = k = |\sigma|-|\pi|$ and we have the result, or until $\sigma^{(j)}$ is exceptional. Then $\pi$ is exceptional of the same type as $\sigma^{(j)}$. Then we set $m = j$ and $k = m+(|\sigma^{(j)}|-|\pi|)/2$, and $\sigma^{(i)}$ for $j \leq i \leq k$ are exceptional permutations of the same type as $\pi$ and size between $|\pi|$ et $|\sigma^{(j)}|$.
\end{proof}

Note that the paths in $G$ between two simple permutations can be of different length. As example with $\sigma = 5263714$ and $\pi = 3142$, we have a path of length 3 (by $526314$ and $42613$, non exceptional, see Figure~\ref{fig:path3}) and a path of length 2 (by $415263$, exceptional, see Figure~\ref{fig:path2}).
\begin{figure}[ht]
\begin{center} 
\begin{tikzpicture}[scale=.2]
\permutation{5,2,6,3,7,1,4}
\end{tikzpicture}
\hspace{0.5cm}
\begin{tikzpicture}[scale=.2]
\permutation{5,2,6,3,1,4}
\end{tikzpicture}
\hspace{0.5cm}
\begin{tikzpicture}[scale=.2]
\permutation{4,2,5,1,3}
\end{tikzpicture}
\hspace{0.5cm}
\begin{tikzpicture}[scale=.2]
\permutation{3,1,4,2}
\end{tikzpicture}
\caption{Path of lenght 3 from $\sigma = 5263714$ to $\pi = 3142$.}\label{fig:path3}
\end{center}
\end{figure}
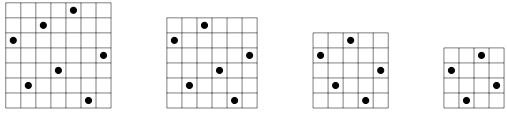
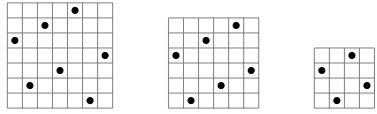
\begin{figure}[ht]
\begin{center} 
\begin{tikzpicture}[scale=.2]
\permutation{5,2,6,3,7,1,4}
\end{tikzpicture}
\hspace{0.5cm}
\begin{tikzpicture}[scale=.2]
\permutation{4,1,5,2,6,3}
\end{tikzpicture}
\hspace{0.5cm}
\begin{tikzpicture}[scale=.2]
\permutation{3,1,4,2}
\end{tikzpicture}
\caption{Path of lenght 2 from $\sigma = 5263714$ to $\pi = 3142$.}\label{fig:path2}
\end{center}
\end{figure}

But if $\sigma$ and $\pi$ are non exceptional, all path from $\sigma$ to $\pi$ have a length $|\sigma|-|\pi|$, and if $\sigma$ and $\pi$ are exceptional, all path from $\sigma$ to $\pi$ have a length $(|\sigma|-|\pi|)/2$.

The last case is $\sigma$ not exceptional and $\pi$ exceptional. Theorem~\ref{thm:1to1} prove that we can always choose a path with only one exceptional permutation: $\pi$.

\begin{prop}
\label{prop:motif-excep}
Let $\pi$ be an exceptional permutation of size $2(n+1)$ where $n \geq 2$, $P$ be a set of $n$ points of $\pi$ and $\pi'$ the exceptional permutation of size $2n$ of the same type as $\pi$. Then there exists a pattern $\pi'$ in $\pi$ which contain all points of $P$.
\end{prop}

\begin{proof}
We have only to prove the result for exceptional permutations of type $3$, the result for the other types follows by symmetry. Suppose that $\pi$ is of type $3$. Then by deleting two points of $\pi$ of consecutive indices, we obtain a pattern $\pi'$. So we have only to prove that there exists two points of consecutive indices which are not in $P$. If it doesn't exist two points of consecutive indices which are not in $P$, then $P$ contains at least half of points of $\pi$ i.e. $n+1$ points, a contradiction.
\end{proof}

\begin{prop}
\label{propF}
Let $\pi' \prec \pi$ be exceptional permutations with $|\pi'|=|\pi|-2$ and $\sigma$ a non exceptional simple permutation such that $\pi \prec \sigma$ and $|\sigma| = |\pi|+1$. Then there exists a simple permutation $\tau$ such that $\pi' \prec \tau \prec \sigma$ and $|\tau| = |\pi'|+1$.
\end{prop}

\begin{proof}
Let $n = |\sigma| = |\pi|+1$, as $\pi \prec \sigma$ there exists an index $k \in \{1 \dots n\}$ such that we obtain $\pi$ from $\sigma$ by deleting $\sigma_k$. As $\sigma$ is simple, we can't have ($k \in \{1, n\}$ and $\sigma_k \in \{1, n\}$). So there exist $i$ and $j$ in $\{1 \dots n\}$ such that $i = k-1$ and $j = k+1$, or $\sigma_i = \sigma_k-1$ and $\sigma_j = \sigma_k+1$. As $\sigma$ is simple, There exists a point $\sigma_{i'}$ which separates $\sigma_i$ from $\sigma_k$ and a point $\sigma_{j'}$ which separates $\sigma_j$ from $\sigma_k$.

If $|\pi| \geq 10$, from Proposition~\ref{prop:motif-excep} there exists a pattern $\pi'$ in $\pi$ (thus in $\sigma$) which contains $\sigma_i$, $\sigma_j$, $\sigma_{i'}$ and $\sigma_{j'}$. Let $\tau$ be the permutation obtained from this pattern $\pi'$ in $\sigma$ and point $\sigma_k$. Then $\pi' \prec \tau \prec \sigma$ and $|\tau| = |\pi'|+1$. Moreover if $\tau$ were not simple, as $\pi' = \tau \setminus \{\sigma_k\}$ is simple, from Proposition~\ref{prop:intervalleOuCoin} $\sigma_k$ belongs to an interval $I$ of size $2$ of $\tau$ ($\sigma_k$ is not in a corner of the graphical representation of $\tau$ because $i < k < j$ or $\sigma_i < \sigma_k < \sigma_j$). Set $I = \{\sigma_{\ell}, \sigma_k\}$, then $\ell =i$ or $\ell =j$, excluded because $\sigma_{i'}$ separates $\sigma_i$ from $\sigma_k$ and $\sigma_{j'}$ separate $\sigma_j$ from $\sigma_k$. Thus $\tau$ is simple and we have the result expected.

As $\pi'$ is exceptional, $|\pi'| \geq 4$ so it remains only to prove the result when $|\pi|=6$ or $|\pi|=8$. We can prove by exhaustive verification that in this case there also exists a set $P$ of points of $\pi$ building a pattern $\pi'$ such that $\tau = P \cup \sigma_k$ fulfil our proposition.

\end{proof}

%

\begin{thm}\label{thm:1to1}
Let $\sigma \not= \pi$ be two simple permutations, $\sigma$ non exceptional and $\ell = |\sigma|-|\pi|$. If $\pi \preceq \sigma$ and $|\pi| \geq 3$ then there exists a chain of simple permutations $\sigma^{(0)} = \sigma, \sigma^{(1)}, \ldots, \sigma^{(\ell-1)}, \sigma^{(\ell)}=\pi$ such that $\forall i$, $\sigma^{(i)} \preceq \sigma^{(i-1)}$ and $|\sigma^{(i-1)}| - |\sigma^{(i)}| = 1$. 
\end{thm}

\begin{proof}
Let $(\sigma^{(i)})_{0\leq i\leq k}$ be a chain of simple permutations given by Theorem~\ref{thm:chain} with $m$ maximum. If $m < k$ then $\sigma^{(m)}$ and $\sigma^{(m+1)}$ are exceptional, and $\sigma^{(m-1)}$ is not exceptional ($m > 0$ because $\sigma$ is not exceptional). We can apply Proposition~\ref{propF}, and we have a simple permutation $\tau$ such that $\sigma^{(m+1)} \preceq \tau \preceq \sigma^{(m-1)}$ and $|\tau| = |\sigma^{(m+1)}|+1$. By Theorem~\ref{thm:theorem0.3}, we have a simple permutation $\rho$ such that $\tau \preceq \rho \preceq \sigma^{(m-1)}$. Then we set $\pi^{(i)} = \sigma^{(i)}$ if $1 \leq i \leq m-1$, $\pi^{(m)} = \rho$, $\pi^{(m+1)} = \tau$ and $\pi^{(i)} = \sigma^{(i-1)}$ if $m+2 \leq i \leq k+1$. Then we have $|\pi^{(i-1)}| - |\pi^{(i)}| = 1$ if $1 \leq i \leq m+1$, which give us a chain of simple permutations verifying conditions of Theorem~\ref{thm:chain} so $m$ is not maximum, a contradiction. Thus $m = k$ and we have the result expected.
\end{proof}

\begin{center}
\begin{figure}[H]\tiny
\begin{tikzpicture}[xscale=1.37,line width=1pt, drop shadow]
\tikzstyle{every node}=[draw,fill=white,rectangle,rounded corners=2,drop shadow,line width=0.5pt];
\tikzstyle{exceptional}=[fill=red!20, drop shadow];
\draw (0,8) node (2 7 4 8 1 6 3 5 ) {$2\,7\,4\,8\,1\,6\,3\,5\,$};
\draw (-2.0,7) node (2 4 7 1 6 3 5 ) {$2\,4\,7\,1\,6\,3\,5\,$};
\draw (2 4 7 1 6 3 5 .north) --  (2 7 4 8 1 6 3 5 .south);
\draw (-1.0,7) node (2 6 4 7 1 3 5 ) {$2\,6\,4\,7\,1\,3\,5\,$};
\draw (2 6 4 7 1 3 5 .north) --  (2 7 4 8 1 6 3 5 .south);
\draw (0.0,7) node (2 6 4 7 1 5 3 ) {$2\,6\,4\,7\,1\,5\,3\,$};
\draw (2 6 4 7 1 5 3 .north) --  (2 7 4 8 1 6 3 5 .south);
\draw (1.0,7) node (2 7 4 1 6 3 5 ) {$2\,7\,4\,1\,6\,3\,5\,$};
\draw (2 7 4 1 6 3 5 .north) --  (2 7 4 8 1 6 3 5 .south);
\draw (2.0,7) node (6 3 7 1 5 2 4 ) {$6\,3\,7\,1\,5\,2\,4\,$};
\draw (6 3 7 1 5 2 4 .north) --  (2 7 4 8 1 6 3 5 .south);
\draw (-4.0,6) node (2 4 1 6 3 5 ) {$2\,4\,1\,6\,3\,5\,$};
\draw (2 4 1 6 3 5 .north) --  (2 4 7 1 6 3 5 .south);
\draw (2 4 1 6 3 5 .north) --  (2 7 4 1 6 3 5 .south);
\draw (-3.0,6) node[exceptional] (2 4 6 1 3 5 ) {$2\,4\,6\,1\,3\,5\,$};
\draw (2 4 6 1 3 5 .north) --  (2 4 7 1 6 3 5 .south);
\draw (2 4 6 1 3 5 .north) --  (2 6 4 7 1 3 5 .south);
\draw (-2.0,6) node (2 4 6 1 5 3 ) {$2\,4\,6\,1\,5\,3\,$};
\draw (2 4 6 1 5 3 .north) --  (2 4 7 1 6 3 5 .south);
\draw (2 4 6 1 5 3 .north) --  (2 6 4 7 1 5 3 .south);
\draw (-1.0,6) node (2 5 3 6 1 4 ) {$2\,5\,3\,6\,1\,4\,$};
\draw (2 5 3 6 1 4 .north) --  (2 6 4 7 1 3 5 .south);
\draw (2 5 3 6 1 4 .north) --  (2 6 4 7 1 5 3 .south);
\draw (0.0,6) node (2 6 4 1 3 5 ) {$2\,6\,4\,1\,3\,5\,$};
\draw (2 6 4 1 3 5 .north) --  (2 6 4 7 1 3 5 .south);
\draw (2 6 4 1 3 5 .north) --  (2 7 4 1 6 3 5 .south);
\draw (1.0,6) node (2 6 4 1 5 3 ) {$2\,6\,4\,1\,5\,3\,$};
\draw (2 6 4 1 5 3 .north) --  (2 6 4 7 1 5 3 .south);
\draw (2 6 4 1 5 3 .north) --  (2 7 4 1 6 3 5 .south);
\draw (2.0,6) node (3 6 1 5 2 4 ) {$3\,6\,1\,5\,2\,4\,$};
\draw (3 6 1 5 2 4 .north) --  (2 4 7 1 6 3 5 .south);
\draw (3 6 1 5 2 4 .north) --  (6 3 7 1 5 2 4 .south);
\draw (3.0,6) node (5 2 6 4 1 3 ) {$5\,2\,6\,4\,1\,3\,$};
\draw (5 2 6 4 1 3 .north) --  (6 3 7 1 5 2 4 .south);
\draw (4.0,6) node (5 3 6 1 4 2 ) {$5\,3\,6\,1\,4\,2\,$};
\draw (5 3 6 1 4 2 .north) --  (2 6 4 7 1 5 3 .south);
\draw (5 3 6 1 4 2 .north) --  (6 3 7 1 5 2 4 .south);
\draw (-2.0,5) node (2 4 1 5 3 ) {$2\,4\,1\,5\,3\,$};
\draw (2 4 1 5 3 .north) --  (2 4 1 6 3 5 .south);
\draw (2 4 1 5 3 .north) --  (2 4 6 1 5 3 .south);
\draw (2 4 1 5 3 .north) --  (2 6 4 1 5 3 .south);
\draw (-1.0,5) node (2 5 3 1 4 ) {$2\,5\,3\,1\,4\,$};
\draw (2 5 3 1 4 .north) --  (2 5 3 6 1 4 .south);
\draw (2 5 3 1 4 .north) --  (2 6 4 1 3 5 .south);
\draw (2 5 3 1 4 .north) --  (2 6 4 1 5 3 .south);
\draw (0.0,5) node (3 1 5 2 4 ) {$3\,1\,5\,2\,4\,$};
\draw (3 1 5 2 4 .north) --  (2 4 1 6 3 5 .south);
\draw (3 1 5 2 4 .north) --  (3 6 1 5 2 4 .south);
\draw (1.0,5) node (3 5 1 4 2 ) {$3\,5\,1\,4\,2\,$};
\draw (3 5 1 4 2 .north) --  (2 4 6 1 5 3 .south);
\draw (3 5 1 4 2 .north) --  (3 6 1 5 2 4 .south);
\draw (3 5 1 4 2 .north) --  (5 3 6 1 4 2 .south);
\draw (2.0,5) node (4 2 5 1 3 ) {$4\,2\,5\,1\,3\,$};
\draw (4 2 5 1 3 .north) --  (2 5 3 6 1 4 .south);
\draw (4 2 5 1 3 .north) --  (5 2 6 4 1 3 .south);
\draw (4 2 5 1 3 .north) --  (5 3 6 1 4 2 .south);
\draw (-0.5,4) node[exceptional] (2 4 1 3 ) {$2\,4\,1\,3\,$};
\draw (2 4 1 3 .north) --  (2 4 1 5 3 .south);
\draw (2 4 1 3 .north) --  (2 5 3 1 4 .south);
\draw (2 4 1 3 .north) --  (3 1 5 2 4 .south);
\draw (2 4 1 3 .north) --  (3 5 1 4 2 .south);
\draw (2 4 1 3 .north) --  (4 2 5 1 3 .south);
\draw (0.5,4) node[exceptional] (3 1 4 2 ) {$3\,1\,4\,2\,$};
\draw (3 1 4 2 .north) --  (2 4 1 5 3 .south);
\draw (3 1 4 2 .north) --  (3 1 5 2 4 .south);
\draw (3 1 4 2 .north) --  (3 5 1 4 2 .south);
\draw (3 1 4 2 .north) --  (4 2 5 1 3 .south);
\draw [red,dotted] (2 4 6 1 3 5 .south) .. controls +(0,-1) and +(-1,0) .. (2 4 1 3 .west);
\end{tikzpicture}
\caption{Poset of simple permutations which are pattern of $2\,7\,4\,8\,1\,6\,3\,5$. Colored node are exceptional ones.}\label{fig:poset}
\end{figure}
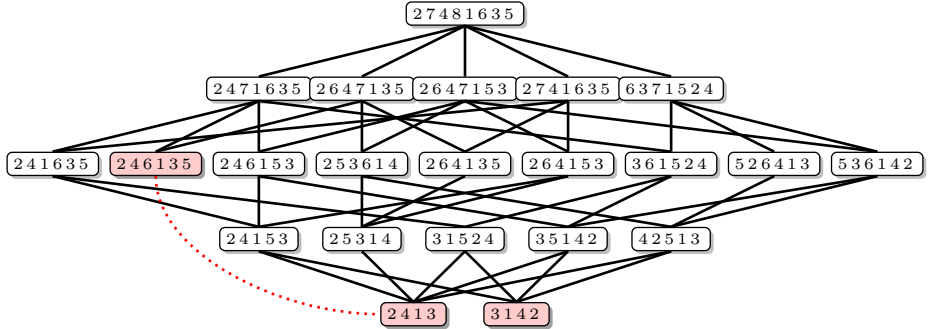
\end{center}

\subsection{Degree of vertices in the poset}

The preceding section proves that the poset is ranked if we omit exceptionnal permutations meaning that each level of the poset corresponds to simple permutations of given size and there exist only edges between permutations of contiguous ranks. In this section, we study the possible edges between two contiguous levels which  provide statistics on simple permutations. More precisely, Proposition~\ref{prop:size-1} proves that if $\sigma$ is a non-exceptional simple permutation, then it exists a simple permutation $\sigma'$ of size $|\sigma|-1$ such that $\sigma' \prec \sigma$. In other words, there exists a point $\sigma_i$ of $\sigma$ such that the permutation obtained when deleting $\sigma_i$ and renormalizing is simple. But how many points of $\sigma$ have this property? To answer this question, we  look at the -possibly multi-graph where vertices are simple permutations and as many edges 
between $\sigma$ and $\sigma'$ as the number of possible ways to insert an element in $\sigma'$ to obtain $\sigma$. Proving that this indeed is a graph shows that our question is equivalent to counting the number of edges between two consecutive levels in the original poset.

%


\begin{prop}
\label{simple->simple}
Let $\sigma = \sigma_{1}\sigma_{2}\dots \sigma_{n}$ be a simple permutation and $\tau$ be a simple permutation of size $n+1$ such that $\sigma \preceq \tau$. Then there exists only one way to obtain $\tau$ by adding a point in $\sigma$.
\end{prop}

\begin{proof}
Suppose that there exist at least 2 ways to do it. Then there are integers $a \neq b$ and $i<k$ such that:\\
$\tau = \sigma'_{1} \dots \sigma'_{i-1}\ a\ \sigma'_{i} \dots \sigma'_{k-1}\ \sigma'_{k} \dots \sigma'_{n}$ and\\
$\tau = \sigma''_{1}\dots\sigma''_{i-1}\ \sigma''_{i}\dots \sigma''_{k-1}\ b\ \sigma''_{k}\dots\sigma''_{n}$\\
where $\sigma'_{j}=\left\{
          \begin{array}{ll}
            \sigma_{j} & \mathrm{if}\ \sigma_{j} < a\\
            \sigma_{j}+1 & \mathrm{otherwise} \\
          \end{array}
        \right.$
and $\sigma''_{j}=\left\{
          \begin{array}{ll}
            \sigma_{j} & \mathrm{if}\ \sigma_{j} < b\\
            \sigma_{j}+1 & \mathrm{otherwise.} \\
          \end{array}
        \right.$

In particular the equality between these two ways to write $\tau$ implies that if $i < k-1$, then $\sigma'_{i} = \tau_{i+1} = \sigma''_{i+1}$, which is impossible because $\sigma$ is simple so $|\sigma_{i}-\sigma_{i+1}| \geq 2$. Thus $i = k-1$, but then the equality implies that $a = \tau_i =\sigma''_i$ and $b = \tau_{i+1} =\sigma'_i$, so $\{a, b\} = \{\sigma_i, \sigma_i+1\}$, which is impossible because $\tau$ is simple. Consequently there is only one way to write $\tau$ from $\sigma$.
\end{proof}

Recalling that $G$ is the graph representing the poset of simple permutations defined at the beginning of Section~\ref{sec:poset}, we consider the graph $G_1$ obtained from $G$ by deleting edges between two exceptional permutations: note that there is an edge in $G_1$ from a simple permutation $\sigma$ to a simple permutation $\pi$ if and only if we can obtain $\pi$ from $\sigma$ by deleting one point.

\begin{defn}
Let $\pi$ be a simple permutation. We define the set of parents $S_{\pi+}$ -resp. children   $S_{\pi-}$ - of $\pi$ in $G_1$ by:\\
$S_{\pi+} = \{\sigma\ |\ \sigma$ is simple, $\pi \preceq \sigma$ and $|\sigma| = |\pi|+1\}$ and\\
$S_{\pi-} = \{\sigma\ |\ \sigma$ is simple, $\sigma \preceq \pi$ and $|\sigma| = |\pi|-1\}$
\end{defn}

\begin{prop}
\label{indegree}
Let $\pi$ be a simple permutation of size $n$. \\Then ${|S_{\pi+}| = (n+1)(n-3)}$.
\end{prop}
\begin{proof}
Permutations of $S_{\pi+}$ are simple permutations obtained from $\pi$ by adding one point. There are $(n+1)^2$ ways to insert a point in $\pi$ (giving permutations not necessarily different): if we consider the graphical representation of $\pi$ in a grid, adding one point to $\pi$ corresponds to choosing one point in the grid, which is of size $(n+1)^2$. But we want only simple permutations, that exclude $4(n+1)$ points in the grid: for one fixed point $\pi_i$ of $\pi$, we can't take one of the $4$ corners of the cell where it is, that exclude $4n$ points which are all different because $\pi$ is simple so there is no points in contiguous cells. And we can't take one of the $4$ corners of the grid, and these $4$ points have not been excluded yet because $\pi$ is simple so there is no point in a corner. There are $4(n+1)$ points excluded among $(n+1)^2$ possibilities and Proposition~\ref{prop:intervalleOuCoin} ensure that they are the only points to exclude, so we have $(n+1)(n-3)$ points left which give simple permutations. We have now to ensure that we can't have the same simple permutations from two different points, which is given by Proposition~\ref{simple->simple}.
\end{proof}

So $|S_{\pi+}|$, which is the indegree of $\pi$ in the graph $G_1$, is independant of $\pi$. We are now interested in $|S_{\pi-}|$, the outdegree of $\pi$ in $G_1$. We know that it depends on $\pi$, and especially that $|S_{\pi-}| = 0$ if and only if $\pi$ is exceptional. We know also that $|S_{\pi-}| \leq |\pi|$. We consider the average outdegree in $G_1$.

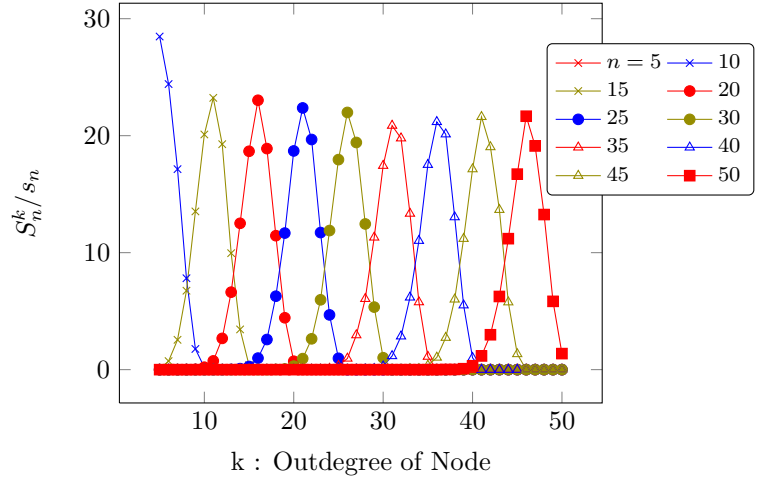
\begin{figure}[H]
\begin{center}
\pgfplotsset{width=8cm,legend style={
at={(1.1,0.9)}, anchor=north, legend columns=2, cells={anchor=west}, font=\footnotesize, rounded corners=2pt,
}}
\begin{tikzpicture}
\begin{axis}[xlabel=k : Outdegree of Node,
ylabel=$S_n^k/s_n$,axis on top] 
\addplot[color=red,mark=x] file {ficDistrib_5.txt};
\addplot[color=blue,mark=x] file {ficDistrib_10.txt};
\addplot[color=olive,mark=x] file {ficDistrib_15.txt};
\addplot[color=red,mark=*] file {ficDistrib_20.txt};
\addplot[color=blue,mark=*] file {ficDistrib_25.txt};
\addplot[color=olive,mark=*] file {ficDistrib_30.txt};
\addplot[color=red,mark=triangle] file {ficDistrib_35.txt};
\addplot[color=blue,mark=triangle] file {ficDistrib_40.txt};
\addplot[color=olive,mark=triangle] file {ficDistrib_45.txt};
\addplot[color=red,mark=square*] file {ficDistrib_50.txt};
\legend{$n=5$,$10$,$15$,$20$,$25$,$30$,$35$,$40$,$45$,$50$}
\end{axis} 
\end{tikzpicture}
\caption{Proportion $S_n^k/s_n$ of simple permutations with outdegree $k$ in $G_1$ among simple permutations of size $n$}\label{fig:outdegree}
\end{center}
\end{figure}

\begin{prop}\label{prop:n4}
Let $D_n$ be the average outdegree of simple permutations of size $n$ in $G_1$. Then $D_n = n-4-\frac{4}{n} + O(\frac{1}{n^2})$.
\end{prop}

\begin{proof}
In $G_1$, there is an edge from a simple permutation $\sigma$ to a simple permutation $\pi$ if and only if we can obtain $\pi$ from $\sigma$ by deleting one point. So edges that come from permutations of size $n$ are those which go to permutations of size $n-1$. Let $s_n$ be the number of simple permutations of size $n$. From Proposition~\ref{indegree} we know that there are $s_{n-1} \times n(n-4)$ such edges. So $D_n = \dfrac{s_{n-1} \times n(n-4)}{s_n}$. But from Theorem 5 of \cite{AAK03} we know that $s_n = \dfrac{n!}{e^2}\Big(1-\dfrac{4}{n}+\dfrac{2}{n(n-1)}+O(n^{-3})\Big)$ and a straightforward computation allows us to conclude.
\end{proof}

We are now interested in the number $S_n^k$ of simple permutations of size $n$ and of outdegree $k$ fixed. As example we know that $S_n^0 = 4$ for every even $n$ and $S_n^0 = 0$ for every odd $n$ (number of exceptional permutations). Figure~\ref{fig:outdegree} show the percentage of simple permutations which have outdegree $k$. Each plot shows this distribution for a given size of permutations as indicated in the  caption. Notice that these plots illustrate the result given in Proposition~\ref{prop:n4}.

\begin{prop}
Let $S_n^k = |\{\pi \ |\ |S_{\pi-}| = k\}|$ be the number of simple permutations of size $n$ and of outdegree $k$ in $G_1$. Then for every fixed $k$, the proportion $\dfrac{S_n^k}{s_n}$ of simple permutations of outdegree $k$ among simple permutations of size $n$ tends to zero when $n$ tends to infinity.
\end{prop}

\begin{proof}
By definition $s_n \times D_n = \sum_{i=0}^{n} i\times S_n^i$. Suppose that there exists $k$ such that $\dfrac{S_n^k}{s_n}$ does not tend to zero, then there exists $\epsilon > 0$ such that $\forall n_0$, $\exists n \geq n_0$ such that $\dfrac{S_n^k}{s_n} > \epsilon$. But then
\[D_n = \sum_{i=0}^{n} i \times \dfrac{S_n^i}{s_n}
= k \times \dfrac{S_n^k}{s_n} + \sum_{i\neq k, i=0}^{n} i \times \dfrac{S_n^i}{s_n}
\leq k + \sum_{i\neq k, i=0}^{n} n \times \dfrac{S_n^i}{s_n}
= k + n\Big(1-\dfrac{S_n^k}{s_n}\Big)
\leq k + n(1-\epsilon)
\]
but from Proposition~\ref{prop:n4}, for $n_0$ large enough $D_n \geq n-5$, a contradiction.
\end{proof}


\section{An algorithm to generate simple permutations in a wreath-closed permutation class}\label{sec:algo}

Theorem~\ref{thm:main} characterizes the pattern relation between non exceptional simple permutations. This theorem ensures that if $\sigma$ is a non exceptional simple permutation and $Av(B)$ a wreath-closed class of permutations, then $\sigma$ does not belong to $Av(B)$
\begin{itemize}
\item if and only if it contains a permutation of $B$ as a pattern.
\item if and only if it is  equal to a permutation of $B$ or contains as a pattern a simple permutation of size $|\sigma|-1$ which does not belong to $Av(B)$. 
\end{itemize}

This recursive test leads to Algorithm~\ref{alg:algo1} (see p.17).
\begin{figure}[ht]
\begin{algorithm}[H] 

\KwData{$B$ a finite set of simple permutations not containing $12$ or $21$} 
\KwResult{$Si$ the set of simple permutations in $Av(B)$ }
$Si_1 \leftarrow \{1\}, Si_2 \leftarrow \{ 12,21\}, Si_3 \leftarrow \emptyset, Si_4\leftarrow \{ 2413,3142\} \setminus B$\; 
$n \leftarrow 5$\;
\While{$Si_{n-1} \not= \emptyset$ or $Si_{n-2} \not= \emptyset$}{
	$Si_n \leftarrow \emptyset$\;
	\For{$\pi \in Si_{n-1}$}{
		\For{ each admissible way to insert a point into $\pi$ and obtain a simple permutation $\sigma$  }{
			\If{$\sigma \not\in B$}{
				inS = true\;
				\For{ each $n$ ways to delete a point in $\sigma$}{
 					Compute the obtained permutation $\tau$\;
					\If{ $\tau$ is simple}{
						\If{$\tau \not\in Si_{n-1} $}{
							inS = false \;
						}
					}
				}
				\If{inS = true}{
					$Si_n \leftarrow Si_n \bigcup \sigma$
				}
			} 
		}
	}
	\For{$\pi$ exceptional of type $i$ $ \in Si_{n-2}$}{
		\If{$\sigma$ exceptional of type $i$ and of size $n \not\in B$}{
			$Si_n \leftarrow Si_n \bigcup \sigma$
		}
	}
	$n \leftarrow n+1$\; 
	 
}
\caption{Generating simple permutations in a wreath-closed class of permutations}\label{alg:algo1}
\end{algorithm}
\end{figure}
Its validity  is proved in the next Theorem based on results from Section~\ref{sec:MotifsSimples}. Note that in order to avoid trivial cases, we assume that $B$ does not contain $12$ or $21$.

\begin{thm} \label{th:algo}
The set $Si_n$ computed by Algorithm~\ref{alg:algo1} is the set of simple permutations of size $n$ contained in $Av(B)$.
\end{thm}

\begin{proof}
The preceding theorem holds for $n \leq 4$. For $n \geq 5$, we show it by induction.

We have to prove that every simple permutation $\sigma$ of size $n$ in $Av(B)$ belongs to $Si_n$. Let $\sigma$ be a simple permutation of size $n$ in $Av(B)$.
If $\sigma$ is not exceptional, there exists $\pi$ simple such that $\pi \preceq \sigma$ and $|\pi| = |\sigma|-1$ (Proposition~\ref{prop:size-1}). 
By induction hypothesis, $\pi \in Si_{n-1}$ so that $\sigma$ is considered at line~$6$ of our algorithm. As $\sigma \in Av(B)$, $\sigma \not\in B$ and every simple pattern $\tau$ of $\sigma$ of length $n-1$ is in $Av(B)$ and by induction hypothesis lies in $Si_{n-1}$.
Thus  line~$18$ is reached and $\sigma$ is added to $Si_n$.
If $\sigma$ is exceptional, $\sigma$ is considered at line~$24$ of our algorithm and is added to $Si_n$ by induction hypothesis.

Reciprocally, let us prove that every permutation $\sigma \in Si_n$ is a simple permutation of size $n$ of $Av(B)$. If $\sigma \in Si_n$ notice first that $\sigma$ is simple and of size $n$. Suppose now that $\sigma \not\in Av(B)$ then it exists $\pi \in B$ ($\pi$ simple) such that $\pi \preceq \sigma$. We have $\sigma \neq \pi$ otherwise $\sigma \in B$ but there is no permutation of $B$ in $Si_n$ (because of lines $7$ and $24$ of the algorithm). If $\sigma$ is not exceptional, using Theorem~\ref{thm:main}, we can find $\tau$ simple of size $n-1$ such that $\pi \preceq \tau \preceq \sigma$, so $\tau \not\in Av(B)$ and by induction hypothesis $\tau \not\in Si_{n-1}$. But our algorithm tests every pattern of $\sigma$ of size $n-1$ in line~$9$ so $\sigma$ is not added to $Si_n$.
If $\sigma$ is exceptional, then $|\pi|$ is even (Proposition~\ref{prop:except}) so $\pi \preceq \sigma'$ where $\sigma'$ is the exceptional permutation of the same type as $\sigma$ of size $|\sigma| - 2$. By induction hypothesis $\sigma' \notin Si_{n-2}$ so $\sigma$ is not added to $Si_n$ and we have the result.
\end{proof}

\begin{prop}
Algorithm~\ref{alg:algo1} terminates if and only if $Av(B)$ contains only a finite number of simple permutations. In this case it gives all simple permutations in $Av(B)$.
\end{prop}

\begin{proof}
If Algorithm~\ref{alg:algo1} terminates, there exists $n \geq 5$ such that $Av(B)$ contains no simple permutation of size $n-1$ or $n-2$. Suppose that $Av(B)$ contains a simple permutation $\sigma$ of size $k \geq n$, then from Proposition~\ref{prop:exceptional} and Proposition~\ref{prop:simplePattern} $\sigma$ has a simple pattern of size $n-1$ or $n-2$ in $Av(B)$, a contradiction. So $Av(B)$ contains no simple permutation of size greater than $n-2$ and Theorem~\ref{th:algo} ensures that the algorithm gives all simple permutations in $Av(B)$.

Conversely if $Av(B)$ contains only a finite number of simple permutations, let $k$ be the size of the greater simple permutation in $Av(B)$. From Theorem~\ref{th:algo}, the algorithm computes $Si_{k+1}=Si_{k+2}=\emptyset$ and the algorithm terminates.
\end{proof}

Before running Algorithm~\ref{alg:algo1} we can test whether $Av(B)$ contains a finite number of simple permutations in time ${\mathcal O}(n \log n)$ where $n=\sum_{\pi \in B} |\pi|$ thanks to the algorithm given in \cite{BBPR09}.
If $Av(B)$ contains a finite number of simple permutations, Algorithm~\ref{alg:algo1} give all simple permutations in $Av(B)$. If $Av(B)$ contains an infinite number of simple permutations, we can use a modified version of the algorithm to obtain all simple permutations in $Av(B)$ of size less than a fixed integer $k$: it is sufficient to replace in the algorithm ``while $Si_{n-1} \not= \emptyset$ or $Si_{n-2} \not= \emptyset$'' by "for $n \leq k$''.

Let us now evaluate the complexity of our algorithm.

\begin{prop}
The complexity of Algorithm~\ref{alg:algo1} is  ${\mathcal O}\big(\sum_{n=5}^{k+1} n^{4} |Si_{n-1}|\big)$ where $k$ is the size of the longest simple permutation in $Av(B)$.
\end{prop}

\begin{proof}
First, we encode every set of permutations as tries, allowing a linear algorithm to check if a permutation is in the set. The {\em while} loop beginning at line~$3$ is done for $n$ from $5$ to $k+2$. The inner loop beginning at line~$5$ is repeated $|Si_{n-1}|$ times (with $|Si_{k+1}|=0$). The loop of line~$6$ is repeated $n(n-4)$ times (see Proposition~\ref{indegree}). This loop performs the following tests:
\begin{itemize}
\item Compute $\sigma$ : ${\mathcal O}(n)$
\item Test whether $\sigma$ is in $B$ : ${\mathcal O}(n)$ using tries.
\item Loop at line~$9$ is performed $n$ times and perform each time the following operations:
\begin{itemize}
\item Compute $\tau$ : ${\mathcal O}(n)$
\item Test whether $\tau$ is simple : ${\mathcal O}(n)$
\item Test whether $\tau \in S_{n-1}$ : ${\mathcal O}(n)$
\end{itemize}
\item Add if necessary $\sigma$ into $Si_n$ : ${\mathcal O}(n)$ as we use tries.
\end{itemize}

Thus the inner part of loop in line~$5$ has complexity ${\mathcal O}(n(n-4)(n+n+n(n+n+n)+n))$ leading to the claimed result.

Indeed, the exceptional case is easy to implement in ${\mathcal O}(n)$ time as there are at most $4$ exceptional permutations of a given size.

\end{proof}

\section{Concluding remarks}

Theorem~\ref{thm:1to1} gives a structural result on simple permutations poset. It has many implications, two of which are explicited in this article: the first one being the average number of points that can be removed in a simple permutation and remain simple, the second one leading to a polynomial time algorithm for computing the set of simple permutations in a wreath-closed class. For the latter problem, we restrict ourselves to wreath-closed class of permutations and unfortunately cannot apply it directly to general classes. Indeed, to adapt our algorithm to the general case, the part of the algorithm from line~$7$ to line~$20$ can be replaced by testing if $\sigma \in Av(B)$ and in that case adding it to $Si_n$. Unfortunately, there is no efficient algorithm to test if a permutation is in a class $Av(B)$. The only known algorithm is to test if $\sigma$ avoids every permutations in $B$. Thus we have the following proposition:
 
 \begin{prop}
 For every class $Av(B)$ containing a finite number of simple permutations, we can compute the simple permutations in $Av(B)$ in time $\sum_{n=5}^{k+1} |Si_{n-1}|n^2f_{Av(B)}(n)$ where $k$ is the size of the longest simple permutation in $Av(B)$ and $f_{Av(B)}(n)$ the complexity of testing if a permutation of size $n$ belongs to $Av(B)$.
 \end{prop}
Note that in the preceding proposition function $f_{Av(B)}$ is bounded by $\sum_{\tau \in B} n^{|\tau|}$. Indeed a naive algorithm consists in testing the pattern condition for each permutation $\tau$ in the basis. In some cases, this test can be improved, see for example \cite{AAAH01,BRV07}. In general case, this give a complexity of order ${\mathcal O} \big(|B|.|Si_B|.k^{p+2}\big)$ for computing the set $Si_B$ of simple permutations in $Av(B)$, with $p = \max \{|\tau| : \tau \in B\}$ and $k = \max \{|\pi| : \pi \in Si_B\}$. For wreath-closed classes, Algorithm~\ref{alg:algo1} has a complexity of order ${\mathcal O}\big(|Si_B|.k^{4}\big)$

An open question is whether there exists a more efficient algorithm in the general case and more precisely for testing if a permutation belongs to a given class.

\bibliographystyle{plain}
\bibliography{biblio}
\end{document}